\documentclass{birkjour}
\usepackage{amssymb}
\usepackage{enumerate}
\usepackage[noadjust]{cite}

\theoremstyle{theorem}
\newtheorem{thm}{Theorem}

\newtheorem{prop}[thm]{Proposition}
\theoremstyle{definition}

\theoremstyle{remark}
\newtheorem{rem}[thm]{Remark}

\renewcommand{\sc}[2]{\langle #1|#2 \rangle}

\newcommand{\R}{\mathbb R}

\newcommand{\C}{\mathbb C}
\newcommand{\N}{\mathbb N}

\renewcommand{\H}{\mathcal{H}}

\DeclareMathOperator{\Gr}{Gr_{res}}

\newcommand{\ur}{\mathfrak u_{\textrm{res}}(\H)}
\newcommand{\urp}{\mathfrak u_\textrm{res}^1(\H)}

\DeclareMathOperator{\Tr}{Tr}

\DeclareMathOperator{\im}{im}

\renewcommand{\L}{\mathcal{L}(\H)}

\newcommand{\abs}[1]{\left\vert#1\right\vert}

\renewcommand{\to}{\rightarrow}

\newcommand{\be}{\begin{equation}}
\newcommand{\ee}{\end{equation}}

\begin{document}

\title[Integrable system on partial isometries]{Integrable system on partial isometries:\\ a finite dimensional picture}
\author[T.~Goli\'nski]{Tomasz Goli\'nski}
\address{University of Bia\l ystok\\ Cio\l kowskiego 1M\\15-245 Bia\l ystok\\ Poland}
\email{tomaszg@math.uwb.edu.pl}
\author[A.B.~Tumpach]{Alice Barbora Tumpach}
\address{\\
\begin{minipage}{0.5\textwidth}
UMR CNRS 8524\\
UFR de Math\'ematiques\\
Laboratoire Paul Painlev\'e\\
59 655 Villeneuve d'Ascq Cedex\\ France \\
\end{minipage}
\begin{minipage}{0.5\textwidth}
Institut CNRS Pauli\\ UMI  CNRS 2842\\ Oskar-Morgenstern-Platz 1 \\1090 Wien\\Austria \\
\end{minipage}
}
\email{alice-barbora.tumpach@univ-lille.fr}
\thanks{This research was partially supported by joint National Science Centre, Poland (number 2020/01/Y/ST1/00123) and Fonds zur F\"orderung der wissenschaftlichen Forschung, Austria (number I 5015-N) grant ``Banach Poisson--Lie groups and integrable systems''. The authors would like to thank the Erwin Schr\"odinger Institute for its hospitality during the thematic programme ``Geometry beyond Riemann: Curvature and Rigidity''.}

\begin{abstract}
The aim of the paper is to present the integrable systems on partial isometries which are related to the restricted Grassmannian in finite dimensional context. Some explicit solutions are obtained.
\end{abstract}
\subjclass{70H06,34A05,47B99,53D17}
\keywords{partial isometries, Magri method, integrable systems, restricted Grassmannian}

\maketitle

\section{Introduction}

This paper deals with a certain hierarchy of integrable bihamiltonian systems constructed on Banach Lie--Poisson spaces related to the restricted Grassmannian $\Gr$. This hierarchy was introduced in the paper \cite{GO-grass}, and its further properties were studied in \cite{GO-grass2,GT-momentum,GO-grass-bial}. The background on the restricted Grassmannian and related objects can be found e.g. in \cite{segal, wurzbacher, Ratiu-grass, tumpach-bruhat, OR, Oext}.

In essence one deals with infinite dimensional differential geometry to obtain this system, which leads to certain technical difficulties. In the papers \cite{Oext,Ratiu-grass} the structure of Banach Lie--Poisson spaces (the notion was introduced in \cite{OR}) related to the restricted Grassmannian were studied.
In the paper \cite{GO-grass} the hierarchy of Hamiltonian integrable systems was constructed on the Banach Lie--Poisson space $\urp$, which is a predual space to the Banach Lie algebra $\ur$. It turns out that these systems descend to the space $L^2(\H_-,\H_+)$, where $\H_-$ and $\H_+$ are complex Hilbert subspaces (finite or infinite dimensional) of a complex separable Hilbert space $\H=\H_+\oplus\H_-$ and $L^2$ denotes the ideal of Hilbert--Schmidt operators. Subsequently it was demonstrated in \cite{GT-momentum} that under a certain condition this system can be further reformulated on the set of partial isometries acting from the Hilbert space $\H_+$ to $\H_-$. 

However for the sake of this paper we will restrict our attention to the finite dimensional case. It will significantly simplify the situation. We will omit the definition of underlying objects (Banach Lie groups, Banach Lie algebras and Banach Lie--Poisson spaces) and present the system only on the space of $N\times M$ matrices. We illustrate the situation by finding a solution in the case of rank one partial isometries. The description of solution in the case of higher rank partial isometry remains the subject of further study. Another finite dimensional approach to this hierarchy has been studied in \cite{GO-grass2}, where the relationship with multimode nonlinear optical systems was discussed.

\section{Presentation of the commuting equations of motion}
\subsection{Equations on the space of skew-hermitian operators in $\H$}
We begin by considering a polarized Hilbert space, i.e. a Hilbert space $\H$ with a chosen orthogonal decomposition $\H=\H_+\oplus\H_-$. 
In order to obtain a hierarchy of commuting equations of motion we need to introduce the following notation. Let $P_\pm$ be the orthogonal projector on $\H_\pm$ and consider
the space of Hilbert--Schmidt operators $L^2(\H_-,\H_+)$.
However in finite dimensional setting we can identify $L^2(\H_-,\H_+)$ with the set of $N\times M$ complex matrices $Mat_{N\times M}(\C)$, where $\dim\H_-=M$ and $\dim\H_+=N$.

The hamiltonian systems under consideration were first obtained on a certain space of skew-hermitian operators in $\H$ which can be identified in this paper with the set of all skew-hermitian $(N+M)\times(N+M)$ matrices.
A skew-hermitian matrix $\mu$ will be written in the block form consistent with the decomposition of $\H$ as
\begin{equation}
  \mu:=\begin{pmatrix}
         \mu_{++}& -\mu_{-+}^*\\
         \mu_{-+}&\mu_{--}
       \end{pmatrix}\in L(\H),
\end{equation}
where operators $P_+\mu P_+ = \mu_{++}\in L(\H_+)$ and $P_-\mu P_- = \mu_{--}\in L(\H_-)$ are skew-hermitian. 
The following family of homogeneous polynomials was first introduced in \cite{GO-grass}
\begin{equation}\label{def_H}
  H_k^n(\mu):=\!\!\!\!\!\!\!\!\!\sum_{i_0,i_1,\ldots i_n\in\{0,1\} \atop i_0+\ldots +i_n=k}\!\!\!\!\!\!\!\!\! P_+^{i_0}\mu P_+^{i_1}\mu\ldots \mu P_+^{i_n},
\end{equation}
where $H^n_k$ is of the degree $n\in\N$ in the operator variable $\mu$ and degree $k$ in the projector $P_+$, with $k\leq n+1$. 
The following hierarchy of Lax equations is under consideration in this paper
\begin{equation}\label{H-eq}\frac\partial{\partial t_k^n}\mu = i^{n+1}
  [\mu, H_k^n(\mu)],
\end{equation}
where $n\in\N$ and $k=1,\ldots, n+1$. 

\begin{rem}\label{rem}
Even though it is not easily seen directly from the form \eqref{H-eq} of the equations, the flows corresponding to each equation of the hierarchy commute and the underlying geometric structure guarantees that the diagonal blocks $\mu_{++}$ and $\mu_{--}$ of the operator $\mu$ are constant
\begin{equation}\label{diag-const}\frac\partial{\partial t_k^n}\mu_{++} = 0 \qquad \qquad \frac\partial{\partial t_k^n}\mu_{--} = 0,\end{equation}
see e.g. \cite{GO-grass} or \cite{GT-momentum}.
Moreover the traces of the operators $H^n_k$ are integrals of motion with respect to all times.
\end{rem}

\subsection{Equations on the space of partial isometries from $\H_+$ to $\H_-$}

Let us recall that a partial isometry in a Hilbert space is an isometry from the orthogonal complement of its kernel (called the initial space) onto some Hilbert subspace (called the final subspace). It can be alternatively defined by any of the equivalent conditions: 
\begin{itemize}
 \item $u^*uu^*=u^*$, 
 \item $uu^*u=u$,
 \item $u^*u$ is an orthogonal projector (onto the initial space),
 \item $uu^*$ is an orthogonal projector (onto the final space). 
\end{itemize}
Unlike the isometries, they do not constitute a group but a groupoid, which also possesses a structure of Banach Lie groupoid, see \cite{OS,GJS-partiso}.

Since the diagonal blocks $\mu_{++}$ and $\mu_{--}$ of a skew-hermitian matrix $\mu$ are constant along the flows \eqref{H-eq} (see Remark~\ref{rem}),
we will investigate in the remainder of the paper the situation where additionally 
\begin{equation}\label{mupp0}
  \mu_{++}=0.\end{equation} 
In paper \cite{GT-momentum} it was demonstrated that under this assumption the following holds:
\begin{equation}\frac\partial{\partial t^n_k}\,(\mu_{+-}\mu_{-+}) = 0.\end{equation}
Thus the modulus $\abs{\mu_{-+}} = \sqrt{\mu_{+-}\mu_{-+}}$ is constant along the flows for all $t^n_k$, $n\in \N$, $k\leq n+1$. 
In consequence one can consider the polar decomposition of the operator $\mu_{-+}$ of the form
\begin{equation}\mu_{-+} = u B,\end{equation}
where $B:=\abs{\mu_{-+}}$ is fixed and $u:\H_+\to\H_-$ is a partial isometry with the final space $\overline{\im \mu_{-+}}$ and the initial space $(\ker \mu_{-+})^\perp = (\ker B)^\perp$. To simplify the equations let us introduce the notation $D=\mu_{--}$ so that the block decomposition of the matrix $\mu$ is under our assumptions of the form
\begin{equation}\label{block}
  \mu:=\begin{pmatrix}
         0 & -Bu^*\\
         uB&D
       \end{pmatrix}\in L(\H).
\end{equation}

As $\mu$ evolves under one of the equations in the hierarchy \eqref{H-eq}, the time-evolution of $u$ needs not to be continuous in general. However in the finite dimensional case $B$ is always partially invertible (i.e. there exists an operator $C$ such that $BC=u^*u$). By direct calculations using recurrence relations (see \cite{GT-momentum}) we obtain the following system of equations on the partial isometry~$u$:
\begin{equation}\label{eq-partiso}
  \frac\partial{\partial t^n_k} u = i^{n+1}(\mu H^{n-1}_{k-1}(\mu))_{--}\,u.
\end{equation}
Note that since the initial space is fixed, these evolutions preserve the projector $uu^*$.

\begin{prop}
For $k> n/2+1$ the right hand side of equation \eqref{eq-partiso} vanishes. In consequence it is sufficient to restrict our considerations to $t^n_k$ for $k\leq n/2+1$. 
\end{prop}
\begin{proof}
  It follows from the condition \eqref{mupp0} that in the definition
  \eqref{def_H} of $H^{n-1}_{k-1}$ we cannot put too many projectors
  $P_+$ without having the expression
  $$P_+^{i_0}\mu P_+^{i_1}\mu\ldots \mu P_+^{i_n}$$ vanish. Thus at
  least every other $i_j$ should be zero. Moreover, terms in
  $H^{n-1}_{k-1}$ with $P_+$ at the end also vanish due to the
  presence of the projector $P_-$ in the equation
  \eqref{eq-partiso}. If $n-1$ is even then the maximal number of
  operators $P_+$ is $(n-1)/2$ which need to be put between groups of
  $\mu^2$ or at the beginning. Similarly if $n-1$ is odd then we get
  $n/2-1$ places for $P_+$ between groups of $\mu^2$ and we can place
  the remaining $\mu$ at the end gaining one extra place for $P_+$.
\end{proof}

Let us write down explicitly a few of the equations from this
hierarchy. Equations for $n=1$ and $n=2$ are linear:
\begin{align}\frac\partial{\partial t^1_1} u &= -Du,\\
  \frac\partial{\partial t^2_1} u &= i(uB^2-D^2u).
\end{align}
The first non-linear equation in this system is obtained for $n=3$ and $k=1$:
\begin{equation}\label{eq_3_1}\frac\partial{\partial t^3_1} u = -DuB^2-uB^2u^*Du+D^3u
\end{equation}
and for $n=3$ and $k=2$:
\begin{equation}\label{eq_3_2}\frac\partial{\partial t^3_2} u = -DuB^2-uB^2u^*Du.\end{equation}
and for $n=4$ and $k=2$:
\begin{equation}\label{eq_4_2}\frac\partial{\partial t^4_2} u = i(2uB^4-D^2uB^2-uB^2u^*D^2u-DuB^2u^*Du).
\end{equation}
In next sections we present particular solutions to these equations.

\section{2+2 dimensional case}

In this section we consider a toy model with $N=M=2$, so $\H=\C^2\times\C^2$. 
For the sake of simplicity let us take $u:\C^2\to\C^2$ as a partial isometry with one dimensional initial space.
By a change of basis in $\H_+=\C^2$, the initial space can be chosen as $\C\times\{0\}$ and then $u$ is of the form
\begin{equation}u =
  \begin{pmatrix}
    \alpha & 0  \\
    \beta& 0 
  \end{pmatrix}
\end{equation}
for some $\alpha,\beta\in\C$ satifying $\abs\alpha^2+\abs\beta^2 = 1$. Since $B$ is supposed to be a positive $2\times2$ matrix with image equal to the initial space of $u$, it is necessary of the form
\begin{equation}B =
  \begin{pmatrix}
  b & 0  \\
  0 & 0 
  \end{pmatrix},
\end{equation}
where $\R\ni b > 0$.
By choosing appropriate basis in $\H_-=\C^2$ we can always consider the matrix $D$ in the diagonal form
\begin{equation}D =
  \begin{pmatrix}
  d_1 & 0  \\
  0 & d_2 
  \end{pmatrix},
\end{equation}
where $d_1,d_2\in i\R$.

In this setting the matrix equation \eqref{eq_3_1} becomes a system of two scalar equations for unknown complex-valued functions $\alpha$ and $\beta$ depending on the constants $b, d_1, d_2$:
\begin{equation}
  \begin{cases}
    \frac\partial{\partial t^3_1} \alpha = -(b^2 d_1 + b^2 d_1 \abs\alpha^2 + b^2 d_2 \abs\beta^2 -d_1^3)\alpha\\
    \frac\partial{\partial t^3_1} \beta = -(b^2 d_2 + b^2 d_1 \abs\alpha^2 + b^2 d_2 \abs\beta^2 -d_2^3)\beta
  \end{cases},
\end{equation}
together with the constraint $\abs\alpha^2+\abs\beta^2 = 1$. Using this constraint these equations can be easily decoupled and assume the form
\begin{equation}\label{eq_2x2_3_1}
  \begin{cases}
\frac\partial{\partial t^3_1} \alpha = i(A_1 + A_2 \abs\alpha^2)\alpha\\
\frac\partial{\partial t^3_1} \beta = i(A_3 + A_4 \abs\beta^2)\beta
  \end{cases},
\end{equation}
where $A_1 = i(b^2 d_1 + b^2 d_2 -d_1^3)$, $A_2 = i(b^2 d_1 - b^2 d_2)$, $A_3 = i(b^2 d_2 + b^2 d_1 -d_2^3)$, $A_4 = i(-b^2 d_1  + b^2 d_2)$ are real constants.
Since the coefficients on the right hand side of those equations in front of the variable is pure imaginary, they can be easily solved using polar coordinates $\alpha = re^{i\varphi}$. In this way the first equation decomposes into 
\begin{equation}
  \begin{cases}
    \frac\partial{\partial t^3_1} r = 0, \\
    \frac\partial{\partial t^3_1} \varphi = A_1+A_2r^2.
  \end{cases}
\end{equation}
The solution is thus $\alpha = \alpha_0 e^{i(A_1+A_2\abs{\alpha_0}^2)t^3_1}$ and similarly $\beta = \beta_0 e^{i(A_3+A_4\abs{\beta_0}^2)t^3_1}$.
Note that the final space of $u$ is $\C{\alpha\choose\beta}$ and in the general case it changes by rotating $\alpha$ and $\beta$ with different but constant velocities. However in a special case if the final space 
was either $\C\times\{0\}$ or $\{0\}\times\C$, it is preserved by this flow. 

Let us write down explicitly another equation from the hierarchy, for a higher value of $n=4$. In this setting the matrix equation \eqref{eq_4_2} assumes the form:
\begin{equation}
  \begin{cases}
    \frac\partial{\partial t^4_2} \alpha = i(2b^4-b^2 d_1^2 - 2 b^2 d_1^2 \abs\alpha^2 - b^2 (d_1 d_2 +d_2^2 )\abs\beta^2)\alpha\\
    \frac\partial{\partial t^4_2} \beta = i(2b^4-b^2 d_2^2 - 2 b^2 d_2^2 \abs\beta^2 - b^2 (d_1 d_2 +d_1^2 )\abs\alpha^2)\beta
  \end{cases}.
\end{equation}
They are of the same type as equations \eqref{eq_2x2_3_1} and can be solved by the same approach.

% For $t^6_1$:
% \begin{equation}\dot \alpha = b \alpha (d_1 (d_1 (d_1 (d_1 (d_1^2-b^2 \alpha \bar\alpha)-b^2 d_1 \alpha \bar\alpha)-\end{equation}
% $$ - b^2 \alpha \bar\alpha (d_1^2-b^2 \alpha \bar\alpha)+b^4 \beta \bar\beta \alpha \bar\alpha)-b^2 \alpha \bar\alpha (d_1 (d_1^2-b^2 \alpha \bar\alpha)-$$
% $$ - b^2 d_1 \alpha \bar\alpha)-b^2 \beta (-(b^2 d_2 \bar\beta \alpha)-b^2 d_1 \bar\beta \alpha) \bar\alpha)-$$
% $$-b^2 \alpha \bar\alpha (d_1 (d_1 (d_1^2-b^2 \alpha \bar\alpha)-b^2 d_1 \alpha \bar\alpha)-$$
% $$-b^2 \alpha \bar\alpha (d_1^2-b^2 \alpha \bar\alpha)+b^4 \beta \bar\beta \alpha \bar\alpha)-b^2 \beta \bar\alpha (-(b^2 \bar\beta \alpha (d_1^2-b^2 \alpha \bar\alpha))+$$
% $$+ d_2 (-(b^2 d_2 \bar\beta \alpha) -b^2 d_1 \bar\beta \alpha)+b^4 \beta \bar\beta^2 \alpha))$$

Looking at the form of the equations we obtain so far one can formulate the following observation:
\begin{prop}\label{prop_2+2}
The equations for $\alpha$ and $\beta$ with respect to the arbitrary time $t^n_k$ are of the form:
\begin{equation}\label{eq-alpha-beta}
  \begin{cases}
    \frac\partial{\partial t^n_k} \alpha = i p_1(\abs{\alpha}^2,\abs{\beta}^2) \alpha\\
    \frac\partial{\partial t^n_k} \beta = i p_2(\abs{\alpha}^2,\abs{\beta}^2) \beta
  \end{cases},
\end{equation}
where $p_1,p_2$ are polynomials with real coefficients depending on parameters $b$, $d_1$ and $d_2$.
\end{prop}
\begin{proof}
We see that the right hand side of \eqref{eq-partiso} is an operator from $\H_+$ to $\H_-$ and noting that $Bu^*u = B$ we conclude it consists of terms of the form
\begin{multline}\label{gen_term}
i^{n+1} (D^{i_1} u B^{j_1} u^*)\cdot (D^{i_2} u B^{j_2} u^*) \cdot\ldots\cdot (D^{i_l} u B^{j_l} u^*) u 
\\
= i^{n+1} D^{i_1} u B^{j_1} (u^* D^{i_2} u B^{j_2}) \cdot (u^* D^{i_3} u B^{j_3})\cdot\ldots\cdot (u^*D^{i_l} u B^{j_l}) u^* u
\end{multline}
for some $i_1,\ldots i_l, j_1,\ldots j_l\in \{0,1,\ldots\}$.
We can move the parentheses so the matrices form such groups $u^*D^s u B^r$ and observe that they are equal to $b^r(d_1^s \abs\alpha^2+d_2^s\abs\beta^2) u^*u$. In the end we get the expression 
\begin{multline}
  i^{n+1} D^{i_1} u B^{j_1} p(\abs{\alpha}^2,\abs{\beta}^2)u^*u = 
  i^{n+1} b^{j_1} p(\abs{\alpha}^2,\abs{\beta}^2) D^{i_1} u\\
$$=i^{n+1}b^{j_1}p(\abs{\alpha}^2,\abs{\beta}^2)
  \begin{pmatrix}
    d_1^{i_1} & 0  \\
    0 & d_2^{i_1}  
  \end{pmatrix}
  \begin{pmatrix}
    \alpha & 0  \\
    \beta & 0 
  \end{pmatrix}
\end{multline}
for some polynomial $p$, possibly with complex coefficients, where we have used $u u^* u = u$. 
Summing up matrices of that kind we get
\begin{equation}i
  \begin{pmatrix}
    p_1(\abs{\alpha}^2,\abs{\beta}^2) & 0  \\
    0 & p_2(\abs{\alpha}^2,\abs{\beta}^2)  
  \end{pmatrix}
  \begin{pmatrix}
    \alpha & 0  \\
    \beta & 0 
  \end{pmatrix}.
\end{equation}
It remains to show that the polynomials $p_1$ and $p_2$ indeed have real coefficients. Looking back at equation \eqref{eq-partiso} we observe that the matrix in front of the matrix $u$ on the right hand side is skew hermitian. Thus $p_1$ and $p_2$ are real.
\end{proof}

Using the same method that we used to solve equations \eqref{eq_2x2_3_1} we conclude that the equations \eqref{eq-alpha-beta} decouple and  the solution is the following
\begin{align}
  \alpha &= \alpha_0 e^{ip_1(\abs{\alpha_0}^2,1-\abs{\alpha_0}^2)t^n_k},\\
  \beta &= \beta_0 e^{ip_2(1-\abs{\beta_0}^2,\abs{\beta_0}^2)t^n_k}.
\end{align}

\section{Solution of the equations for rank one partial isometries in arbitrary finite dimension}
Let us go back to a more general case, where the dimensions of $\H_+$ and $\H_-$ are arbitrary, but finite. We will however keep the assumption that the rank of the partial isometry $u$ is equal to 1. By changing independently the basis in $\H_+$ and $\H_-$ in an appropriate manner, we can assume again that both $B$ and $D$ are diagonal and that the initial space of $u$ is spanned by the first basis vector of $\H_+$. In consequence the partial isometry $u: \H_+ \rightarrow \H_-$ is of the form:
\begin{equation}u=
  \begin{pmatrix}
    \alpha_1 & 0 & \ldots & 0  \\
    \vdots & \vdots & \ddots & \vdots \\
    \alpha_M & 0 & \ldots & 0 
  \end{pmatrix},
\end{equation}
where $\abs{\alpha_1}^2+\ldots \abs{\alpha_M}^2 = 1$.
One easily observes that  Proposition \ref{prop_2+2} generalizes to this setting.

\begin{prop}\label{prop_arbitrary}
The equations for the evolution of the coefficients $\alpha_1,\ldots \alpha_M$ with respect to the arbitrary time $t^n_k$ are of the form:
\begin{equation}\label{eq-alpha-beta-arbitrary}
  \frac\partial{\partial t^n_k} \alpha_j = i p_{j,k}^n(\abs{\alpha_1}^2,\ldots,\abs{\alpha_M}^2) \alpha_j,
\end{equation}
where $p_{j,k}^n$ are polynomials with real coefficients depending smoothly on the eigenvalues of the matrices $B$ and $D$.
\end{prop}
\begin{proof}
The proof is completely analogous to the proof of Proposition \ref{prop_2+2}. The formula  \eqref{gen_term} is still valid and since $B$ has only one non-zero eigenvalue, the rest of the argument carries over after replacing matrices with higher dimensional ones.
\end{proof}

In this case the condition $\abs{\alpha_1}^2+\ldots +\abs{\alpha_M}^2 = 1$ is not sufficient to decouple the equations if $\dim \H_->2$. It might still be possible to do so by additionally using the integrals of motion $\Tr H^n_k$. However one can still use the same approach as in the 2+2 case to solve the equations without decoupling them first.

\begin{thm}\label{thm}
The solution to \eqref{eq-partiso} for the case of partial isometries of rank one is the following
% \begin{equation}\alpha_j = \alpha_j^0 e^{ip_j(\abs{\alpha_1^0}^2,\ldots,\abs{\alpha_l^0}^2)t^n_k}.\end{equation}
\begin{equation}\label{solution}
  \alpha_j(t_1^1,t_1^2,t_2^2,\ldots) = \alpha_j^0 \exp\left({i\sum\limits_{\tiny{n,k\leq n/2+1}} p_{j,k}^n(\abs{\alpha_1^0}^2,\ldots,\abs{\alpha_M^0}^2)t^n_k}\right),
\end{equation}
where $\alpha_j^0\in\C$ are initial values.
\end{thm}
\begin{proof}
Using the polar form of the coefficients $\alpha_j = r_j e^{i\varphi_j}$ we obtain from Proposition \ref{prop_arbitrary} the equations in the following form
\begin{equation}
  \begin{cases}
    \frac\partial{\partial t^n_k} r_j = 0 \\
    \frac\partial{\partial t^n_k} \varphi_j = p_{j,k}^n(r_1^2,\ldots,r_M^2)
  \end{cases},
\end{equation}
whose solution is clearly \eqref{solution}.
\end{proof}

\section{Conclusion}
In this paper we considered a family of commuting equations of motions given in Lax form by \eqref{H-eq}. The diagonal blocks $\mu_{++}$ and $\mu_{--}$ being preserved by the flows, we investigate in more details the case where $\mu_{++} = 0$ which leads to equations of motion on the space of partial isometries $u$ from $\H_+$ to $\H_-$. 
The case of partial isometries with rank one is completely solved in Theorem~\ref{thm}.
It remains an open problem what happens in a case when $u$ has a higher rank or even when it is unitary. It will be a subject of a separate study.
% \bibliography{../literatura}

\end{document}